\documentclass{article}
\usepackage{authblk}
\usepackage{graphicx} 

\usepackage[top=1in, bottom=1in, left=1in, right=1in]{geometry}

\usepackage[colorinlistoftodos]{todonotes}
\usepackage{url}
\usepackage{hyperref}
\usepackage{natbib}
\usepackage{booktabs}
\usepackage{mathtools}
\usepackage{amsthm}
\newtheorem{remark}{Remark}
\usepackage{subcaption}
\usepackage{comment}
\usepackage{graphicx} 
\usepackage{amsmath,amssymb}
\usepackage{float}
\usepackage{authblk}
\usepackage{amsthm}
\usepackage{amsthm}
\usepackage{tikz-cd}
\usepackage{verbatim}
\usepackage{xcolor}
\usepackage[colorinlistoftodos]{todonotes}
\newtheorem{theorem}{Theorem}
\newtheorem{lemma}[theorem]{Lemma}
\newtheorem{proposition}[theorem]{Proposition}
\newtheorem{corollary}[theorem]{Corollary}

\theoremstyle{definition}
\newtheorem{definition}[theorem]{Definition}
\usepackage{graphicx} 

\title{Basis Rigidity of the AES S-box and Generic Rigidity of Inversion under Affine Transformations}

\author{
Zheng Zhang \thanks{Corresponding author. Department of Mathematics, Towson University, 7800 York Rd, Towson, MD 21204, USA. Email: \texttt{zhengzhang@towson.edu}.}
Na Zhang \thanks{Department of Mathematics, Towson University, 7800 York Rd, Towson, MD 21204, USA. Email: \texttt{nzhang@towson.edu}.}
}

\date{}

\begin{document}

\maketitle

\begin{abstract}
The AES S-box is constructed from finite field inversion followed by a
fixed affine transformation. Since inversion possesses intrinsic
Frobenius symmetries among its coordinate realizations, we study how
these basis symmetries are altered by outer affine transformations.

We first develop a deterministic rigidity criterion for transformed
inversion and apply it to the AES S-box. This shows that the linear part of the AES S-box affine transformation
alone makes the transformed inversion map basis rigid. We then investigate the corresponding generic problem when
the outer invertible linear transformation varies. The existence of a nontrivial linear stabilizer is reduced to a
conjugacy problem for semilinear candidates arising from two sided
linear equivalences of inversion, which we characterize in terms of relative norms and Frobenius
orbits. We also determine the dimensions of the associated centralizer
algebras exactly. These structural results imply that, for a uniformly
chosen outer linear transformation, the probability that the linear stabilizer is nontrivial is bounded by $2^{-\Omega(n^2)}$, with sharper
finite dimensional bounds obtained from the exact conjugacy condition. Computational experiments independently verify the AES rigidity result,
the conjugacy and centralizer formulas, and the finite dimensional
estimates in small dimensions.
\end{abstract}


\section{Introduction}

The AES S-box~\cite{NIST197} is one of the most important nonlinear components in
modern symmetric-key cryptography. Its construction is algebraic. Over
$\mathbb F_{2^8}$, it is obtained from finite field inversion followed
by a fixed affine transformation. This description makes the AES S-box
a natural object for studying the interaction between finite field
structure and coordinate realization.

A finite field map is represented as a vectorial function only after a
basis over $\mathbb F_2$ is chosen. Different ordered bases generally
lead to different coordinate realizations, but distinct bases can
sometimes produce the same coordinate realization. Finite field inversion is a
basic example. Because inversion commutes with Frobenius automorphisms,
its coordinate realizations possess an intrinsic basis redundancy.
In our previous work~\cite{Zhang2026inversion}, we showed that these
Frobenius symmetries account for all such collisions among ordered basis
representations.

The AES S-box is obtained by placing a fixed affine layer on top of this
highly symmetric inversion map. This makes AES a natural first case for asking how an outer
transformation changes the basis rigidity of inversion. The same
question can then be posed more generally. When inversion is followed
by an arbitrary invertible linear map, when can the resulting map have
a nontrivial linear stabilizer, and is this behavior typical or
exceptional?

We study this question under linear changes of
coordinates induced by finite field basis transformations. We consider
both the behavior for a fixed outer linear map, with the AES S-box linear layer
as a distinguished case, and the generic behavior when the outer linear
map varies over $\mathrm{GL}_n(\mathbb F_2)$.
\subsection{Related Work and Scope}

The AES S-box has been studied extensively from the viewpoint of linear
and affine equivalence. Biryukov et al.~\cite{biryukov2003toolbox} developed algorithms for these equivalence
problems and exhibited $2040$ self-equivalence relations for the
Rijndael S-box. Ranea and Preneel~\cite{ranea2021self} further investigated self-equivalence encodings and
centralizer structures arising in white-box implementations.

General two sided affine self-equivalences are broader than the common linear symmetry considered here, since the input and output
transformations may differ and affine translations may also be allowed.
A finite field basis change is more restrictive, since the same linear
coordinate transformation acts simultaneously on both sides. Thus a
function may possess many general self-equivalences while still being
rigid under basis changes.

Different finite field representations of the AES S-box have also been
studied through polynomial, matrix, and basis-dependent descriptions~\cite{rosenthal2003polynomial,liu2005rijndael}. These works describe how
the algebraic form of the S-box changes with the chosen representation.
Our question is different. We ask when two distinct ordered bases induce
exactly the same coordinate realization.

The main algebraic input to our analysis is the classification of
inversion equivalence pairs due to Yuan et al.~\cite{yuan2019inverse}. In characteristic two and for $n\ge4$, their
result identifies the possible semilinear transformations associated
with finite field inversion. This classification determines the semilinear candidates arising from
two sided linear equivalences of inversion, but it does not determine
which of them are compatible with a fixed outer linear transformation,
nor how frequently such compatibility occurs when the outer
transformation varies. These compatibility and counting questions
are central to the present work.

Broader equivalence classes of inversion have also been considered.
For $n\ge5$, K\"olsch~\cite{kolsch2021ccz} showed that every function
CCZ-equivalent to inversion is EA-equivalent to it, and that every
permutation in its CCZ class is affine equivalent to inversion. Such
results concern membership in equivalence classes. By contrast, we fix
an individual transformed inversion map and study its stabilizer under
linear conjugation.

A separate line of work concerns generic stabilizer triviality.
Ishizuka~\cite{ishizuka2026stabilizers} investigates this question for extended-affine stabilizers of vectorial functions over finite fields. Our setting involves a different symmetry and a different probability space. We fix the nonlinear core as finite field inversion and vary only the outer invertible linear transformation. Genericity in the full function space therefore does not directly determine the exceptional proportion within this structured family. Our rigidity bounds are established independently.

In summary, the present paper lies between the theory of
self-equivalences of cryptographic functions and the algebraic structure
of finite field inversion. Our contribution is to determine how the semilinear equivalence
structure of inversion interacts with a fixed or varying outer linear
transformation, with the AES S-box providing the principal motivating
example.

\subsection{Our Contributions}

Our contributions are as follows:

\begin{enumerate}

\item
We derive a deterministic coefficient criterion for basis rigidity of
transformed inversion. Applying it to the AES S-box linear layer shows that the AES S-box is
basis rigid. No nonidentity linear change of basis preserves its
coordinate realization. Moreover, this linear layer alone enforces the
rigidity and the affine constant is not needed.

\item
We characterize exactly which semilinear candidates arising
from two sided linear equivalences of inversion can occur
as linear symmetries of \(A\circ J_n\) for some
invertible outer linear map \(A\). The resulting conjugacy criterion is expressed in terms
of relative norms and Frobenius orbits. We also determine exactly the
dimensions of the corresponding centralizer algebras.

\item
Using the conjugacy criterion and the centralizer dimensions, we obtain
a generic rigidity theorem for transformed inversion. For a uniformly
chosen invertible outer linear transformation, the probability that the linear stabilizer is nontrivial is bounded by
$
2^{-\Omega(n^2)}.
$
We further derive computable finite dimensional refinements by retaining
the exact conjugacy condition and grouping the admissible norm values according to their degrees over
$\mathbb F_2$.

\item
We provide independent computational verification of the theoretical
results. A closure search on complete truth tables confirms the AES
rigidity result without using the semilinear reduction. Exact enumeration
determines the exceptional probabilities for $n=3,4,5$, while direct
matrix computations verify the conjugacy criterion, the centralizer
dimension formula, and the grouped refined bound for all $3{,}537$
nonidentity semilinear candidates in dimensions $3\le n\le8$.
Monte Carlo sampling provides additional finite dimensional comparisons
in dimensions $5\le n\le8$.

\end{enumerate}

\section{Deterministic Rigidity of Transformed Inversion}
This section develops the deterministic framework needed to analyze
basis rigidity for transformed inversion maps. We first relate basis collisions to linear stabilizers, then use
the two sided linear equivalences of finite field inversion to identify
the possible semilinear stabilizer candidates. We next derive
coefficient compatibility conditions for a fixed outer linear
transformation and apply them to the AES S-box. The resulting AES
rigidity theorem serves as a concrete motivating case for the broader
generic theory developed in Section~3.

\subsection{Basis Representations and Linear Stabilizers}
\label{sec:basis-stabilizers}

Let $\mathbb F=\mathbb F_{2^n}$, and fix an ordered reference basis
$E=(e_1,\ldots,e_n)$. For an ordered $\mathbb F_2$-basis
$B=(b_1,\ldots,b_n)$, let $P_B\in\mathrm{GL}_n(\mathbb F_2)$ be defined by
\[
[x]_E=P_B[x]_B.
\]
Let \(F:\mathbb F\to\mathbb F\) be a map.
We identify \(F\) with its coordinate realization in the
reference basis \(E\). Its realization in the basis \(B\) is
\[
F_B=P_B^{-1}FP_B.
\]

\begin{definition}
\label{def:basis-equivalence}
Two ordered $\mathbb F_2$-bases $B$ and $B'$ are said to be
\emph{equivalent with respect to $F$}, written
\[
B\sim_F B',
\]
if
\[
F_B=F_{B'}.
\]
The \emph{linear stabilizer} of $F$ is
\[
\operatorname{Stab}(F)
=
\{T\in\mathrm{GL}_n(\mathbb F_2):TF=FT\}.
\]
\end{definition}

A direct conjugation argument gives
\begin{equation}
\label{equ1}
 F_B=F_{B'}
\quad\Longleftrightarrow\quad
P_{B'}P_B^{-1}\in\operatorname{Stab}(F). \end{equation}

\begin{definition}
\label{def:basis-rigidity}
The map $F$ is called \emph{basis rigid} if
\[
F_B=F_{B'}
\quad\Longrightarrow\quad
B=B'
\]
for all ordered $\mathbb F_2$-bases $B,B'$ of $\mathbb F$.
\end{definition}

By~\eqref{equ1},
\[
F\text{ is basis rigid}
\quad\Longleftrightarrow\quad
\operatorname{Stab}(F)=\{I\}.
\]
Moreover, the equivalence classes of ordered bases are precisely the
cosets of $\operatorname{Stab}(F)$ in $\mathrm{GL}_n(\mathbb F_2)$.
Since ordered $\mathbb F_2$-bases of $\mathbb F$ are in bijection with
$\mathrm{GL}_n(\mathbb F_2)$, the number of distinct basis
representations is
\[
N_{\mathrm{rep}}(F)
=
\frac{|\mathrm{GL}_n(\mathbb F_2)|}
     {|\operatorname{Stab}(F)|}.
\]

\subsection{Stabilizer Reduction for Affine Transformations of Inversion}
\label{sec:postprocessed-inversion}

Let $\mathbb F=\mathbb F_{2^n}$ and define the inversion map
\[
J_n(x)=
\begin{cases}
x^{-1}, & x\neq 0,\\
0, & x=0.
\end{cases}
\]

\begin{definition}
\label{def:transformed-inversion}
For $A\in\mathrm{GL}_n(\mathbb F_2)$ and $c\in\mathbb F$, define
\[
F_{A,c}(x)=A(J_n(x))+c.
\]
We refer to $F_{A,c}$ as a transformed inversion map with an affine
outer layer.
\end{definition}

Let $\sigma(x)=x^2$ denote the Frobenius automorphism, and for
$a\in\mathbb F^\times$ let
\[
M_a(x)=ax.
\]

Although $F_{A,c}$ contains an affine outer layer, a change of basis
still acts through a single linear map on the input and output
coordinate spaces. Hence its basis symmetries are governed by the
linear stabilizer introduced in subsection~\ref{sec:basis-stabilizers}.

For pure inversion, the linear stabilizer is the Frobenius group
\[
\operatorname{Stab}(J_n)
=
\langle\sigma\rangle
=
\{\sigma^r:0\le r<n\}.
\]
Thus inversion has an exact $n$-fold redundancy among ordered basis
representations. After composition with an outer affine transformation, the resulting linear stabilizer need not be a subgroup of
$\operatorname{Stab}(J_n)$. We therefore use the classification of
two sided linear equivalences of inversion. This yields semilinear
candidates of the form $M_a\sigma^r$, which become linear
symmetries of the transformed map precisely when they satisfy the
compatibility conditions derived below.

\begin{theorem}[Yuan et al.~\cite{yuan2019inverse}]
\label{thm:yuan}
Let $n\ge4$ and $T,U\in\mathrm{GL}_n(\mathbb F_2)$. Then
\[
J_nT=UJ_n
\]
if and only if there exist $a\in\mathbb F^\times$ and $0\le r<n$ such
that
\[
T=M_a\sigma^r,
\qquad
U=M_{a^{-1}}\sigma^r.
\]
\end{theorem}

This is the $p=2$ case of Theorem~1 of
Yuan et al.~\cite{yuan2019inverse}, rewritten in our notation and
composition convention.

We now reduce the stabilizer of $F_{A,c}$ to this classification.

\begin{lemma}
\label{lem:reduction}
For $T\in\mathrm{GL}_n(\mathbb F_2)$,
\[
T\in\operatorname{Stab}(F_{A,c})
\]
if and only if
\[
T(c)=c
\qquad\text{and}\qquad
J_nT=(A^{-1}TA)J_n.
\]
\end{lemma}

\begin{proof}
The relation $TF_{A,c}=F_{A,c}T$ is equivalent to
\[
TAJ_n(x)+T(c)=AJ_n(Tx)+c.
\]
Setting $x=0$ gives $T(c)=c$, and substitution then yields
\[
TAJ_n=AJ_nT,
\]
or equivalently
\[
J_nT=(A^{-1}TA)J_n.
\]
The converse follows by reversing the argument.
\end{proof}

\begin{theorem}
\label{thm:stabilizer-reduction}
Let $n\ge4$. Then
\[
\operatorname{Stab}(F_{A,c})
=
\left\{
M_a\sigma^r:
A^{-1}M_a\sigma^rA=M_{a^{-1}}\sigma^r,\;
M_a\sigma^r(c)=c
\right\},
\]
where $a\in\mathbb F^\times$ and $0\le r<n$.
\end{theorem}

\begin{proof}
By Lemma~\ref{lem:reduction}, any
$T\in\operatorname{Stab}(F_{A,c})$ satisfies
\[
T(c)=c,
\qquad
J_nT=(A^{-1}TA)J_n.
\]
Theorem~\ref{thm:yuan} therefore gives
$T=M_a\sigma^r$ and
\[
A^{-1}M_a\sigma^rA=M_{a^{-1}}\sigma^r.
\]
The converse follows from the same two results.
\end{proof}

Thus, for a fixed outer linear map $A$, every linear symmetry must be one of the semilinear candidates
$M_a\sigma^r$ and satisfy a
specific compatibility relation with $A$. In the next subsection, we
translate this relation into explicit coefficient conditions using the
linearized polynomial representation of $A$. For later use, we write
\[
T_{a,r}:=M_a\sigma^r,
\qquad
a\in\mathbb F^\times,\quad 0\le r<n.
\]

\subsection{Coefficient Compatibility and Deterministic Rigidity}
\label{subsec:coefficient-rigidity}

The stabilizer reduction above restricts every linear symmetry of
$A\circ J_n$ to a semilinear map of the form $M_a\sigma^r$. We now
translate the remaining compatibility condition into constraints on the
linearized polynomial coefficients of $A$.

Let
\[
A(x)=\sum_{i=0}^{n-1}\lambda_i x^{2^i},
\qquad
\lambda_i\in\mathbb F,
\]
be the unique linearized polynomial representation of
$A\in\mathrm{GL}_n(\mathbb F_2)$, and write
\[
S_A=\{i:\lambda_i\neq0\}.
\]

\begin{proposition}
\label{prop:coefficient-compatibility}
Let $n\ge4$, and suppose
\[
T_{a,r}=M_a\sigma^r\in\operatorname{Stab}(A\circ J_n),
\]
where $a\in\mathbb F^\times$ and $0\le r<n$. Then, for every
$i\in S_A$,
\[
a^{2^i+1}=\lambda_i^{\,1-2^r}.
\]
Consequently, for every $i,j\in S_A$,
\[
\left(
\frac{\lambda_j^{\,2^i+1}}
     {\lambda_i^{\,2^j+1}}
\right)^{1-2^r}=1.
\]
\end{proposition}

\begin{proof}
By Theorem~\ref{thm:stabilizer-reduction}, with $c=0$, the stabilizing
condition is
\[
M_a\sigma^r A=A M_{a^{-1}}\sigma^r.
\]
Applying both sides to $x\in\mathbb F$ gives
\[
aA(x)^{2^r}=A(a^{-1}x^{2^r}).
\]
Using the linearized polynomial representation of $A$,
\[
\sum_{i=0}^{n-1}
a\lambda_i^{2^r}x^{2^{i+r}}
=
\sum_{i=0}^{n-1}
\lambda_i a^{-2^i}x^{2^{i+r}}.
\]
Since $\sigma^n=I$, the Frobenius exponents are understood modulo $n$.
Uniqueness of the linearized polynomial representation therefore gives
\[
a\lambda_i^{2^r}
=
\lambda_i a^{-2^i}
\]
for every $i$. Hence, whenever $\lambda_i\neq0$,
\[
a^{2^i+1}
=
\lambda_i^{\,1-2^r}.
\]

For $i,j\in S_A$, raising the relation for $i$ to the power
$2^j+1$ and the relation for $j$ to the power $2^i+1$ eliminates $a$
and gives
\[
\lambda_i^{(2^j+1)(1-2^r)}
=
\lambda_j^{(2^i+1)(1-2^r)}.
\]
Rearranging yields
\[
\left(
\frac{\lambda_j^{\,2^i+1}}
     {\lambda_i^{\,2^j+1}}
\right)^{1-2^r}=1.
\]
\end{proof}

For $i,j\in S_A$, define
\[
\Gamma_{ij}
=
\frac{\lambda_j^{\,2^i+1}}
     {\lambda_i^{\,2^j+1}}
\in\mathbb F^\times.
\]

\begin{theorem}
\label{thm:coefficient-rigidity}
Let $n\ge4$. Suppose that the linearized polynomial coefficients of
$A\in\mathrm{GL}_n(\mathbb F_2)$ satisfy the following two conditions:
\begin{enumerate}
\item
\[
\gcd\bigl(2^n-1,\{\,2^i+1:i\in S_A\,\}\bigr)=1;
\]
\item for every $1\le r<n$, there exist $i,j\in S_A$ such that
\[
\operatorname{ord}(\Gamma_{ij})\nmid 2^r-1.
\]
\end{enumerate}
Then
\[
\operatorname{Stab}(A\circ J_n)=\{I\}.
\]
\end{theorem}

\begin{proof}
Let $M_a\sigma^r\in\operatorname{Stab}(A\circ J_n)$.

If $r=0$, Proposition~\ref{prop:coefficient-compatibility} gives
\[
a^{2^i+1}=1
\]
for every $i\in S_A$. Since the order of $a$ also divides $2^n-1$, the first condition forces $a=1$.

Now suppose $1\le r<n$. Proposition~\ref{prop:coefficient-compatibility}
implies
\[
\operatorname{ord}(\Gamma_{ij})\mid 2^r-1
\]
for every $i,j\in S_A$, contradicting the second condition. Hence no
stabilizing map with $r\neq0$ exists.

Therefore the identity is the only element of
$\operatorname{Stab}(A\circ J_n)$.
\end{proof}

\begin{remark}
\label{rem:lambda-zero}
If $\lambda_0\neq0$, the first condition of
Theorem~\ref{thm:coefficient-rigidity} is automatic, since
$2^0+1=2$ whereas $2^n-1$ is odd.
\end{remark}

The criterion isolates the coefficient obstruction responsible for
deterministic rigidity. We now apply it to the AES S-box linear layer, where a single pair of
coefficients is sufficient to exclude every candidate with nonzero
Frobenius exponent.

\subsection{AES Basis Rigidity}
\label{subsec:aes-basis-rigidity}

We now apply Theorem~\ref{thm:coefficient-rigidity} to the
linear part of the AES S-box affine transformation.
Following FIPS~197~\cite{NIST197}, we use the field representation
\[
\mathbb F_{2^8}
=
\mathbb F_2[t]/(t^8+t^4+t^3+t+1)
\]
with the ordered polynomial basis
\[
E=(1,t,\ldots,t^7).
\]
We identify a byte \(\sum_{i=0}^{7}b_i2^i\) with the field
element \(\sum_{i=0}^{7}b_it^i\), where \(b_i\in\mathbb F_2\).
All byte values below are written in hexadecimal.

In this representation, the AES S-box is
\[
S_{\mathrm{AES}}(x)
=
A_{\mathrm{AES}}(J_8(x))+\texttt{63},
\]
where \(A_{\mathrm{AES}}\) is the invertible
\(\mathbb F_2\)-linear map defined by
\[
A_{\mathrm{AES}}\left(\sum_{i=0}^{7}b_it^i\right)
=
\sum_{i=0}^{7}
\bigl(b_i+b_{i+4}+b_{i+5}+b_{i+6}+b_{i+7}\bigr)t^i.
\]
Here the bit indices are taken modulo \(8\), and the coefficients
are added in \(\mathbb F_2\).

The linearized polynomial representation of this map is
\[
A_{\mathrm{AES}}(x)
=
\sum_{i=0}^{7}\lambda_i x^{2^i},
\]
with
\[
(\lambda_0,\ldots,\lambda_7)
=
(\texttt{05},\texttt{09},\texttt{F9},\texttt{25},
 \texttt{F4},\texttt{01},\texttt{B5},\texttt{8F}).
\]
These coefficients can be verified by evaluating both
representations on the basis elements \(1,t,\ldots,t^7\).

\begin{theorem}
\label{thm:aes-linear-rigidity}
The linear part \(A_{\mathrm{AES}}\) of the AES S-box affine
transformation satisfies
\[
\operatorname{Stab}(A_{\mathrm{AES}}\circ J_8)=\{I\}.
\]
In particular, \(A_{\mathrm{AES}}\circ J_8\) is basis rigid.
\end{theorem}

\begin{proof}
Since \(\lambda_0\ne0\), the first condition of
Theorem~\ref{thm:coefficient-rigidity} holds by
Remark~\ref{rem:lambda-zero}.

For the second condition, consider the pair \((i,j)=(0,1)\).
The coefficients above give
\[
\Gamma_{01}
=
\frac{\lambda_1^2}{\lambda_0^3}
=
\texttt{E7}.
\]
Computation in the specified field gives
\[
\texttt{E7}^{15}=\texttt{35},\qquad
\texttt{E7}^{51}=\texttt{50},\qquad
\texttt{E7}^{85}=\texttt{BC}.
\]
The order of \(\Gamma_{01}\) divides
\(\lvert\mathbb F_{2^8}^{\times}\rvert=255\).
Since \(255=3\cdot5\cdot17\) and none of these powers equals \(1\),
we obtain
\[
\operatorname{ord}(\Gamma_{01})=255.
\]
Consequently, for every \(1\le r\le7\),
\[
\operatorname{ord}(\Gamma_{01})\nmid 2^r-1,
\]
because \(0<2^r-1<255\).

Both conditions of Theorem~\ref{thm:coefficient-rigidity}
are therefore satisfied, yielding
\[
\operatorname{Stab}(A_{\mathrm{AES}}\circ J_8)=\{I\}.
\]
\end{proof}

Thus the linear layer alone makes transformed inversion basis
rigid. The same conclusion holds for the full AES S-box.

\begin{corollary}
\label{cor:aes-rigidity}
\[
\operatorname{Stab}(S_{\mathrm{AES}})=\{I\}.
\]
\end{corollary}

\begin{proof}
Let \(T\in\operatorname{Stab}(S_{\mathrm{AES}})\). Then
\[
T\bigl(A_{\mathrm{AES}}(J_8(x))+\texttt{63}\bigr)
=
A_{\mathrm{AES}}(J_8(Tx))+\texttt{63}
\]
for every \(x\in\mathbb F_{2^8}\).
Setting \(x=0\) gives \(T(\texttt{63})=\texttt{63}\).
By linearity of \(T\), the constant terms therefore cancel,
leaving
\[
T A_{\mathrm{AES}}J_8
=
A_{\mathrm{AES}}J_8 T.
\]
Hence
\[
T\in\operatorname{Stab}(A_{\mathrm{AES}}\circ J_8),
\]
and Theorem~\ref{thm:aes-linear-rigidity} implies \(T=I\).
\end{proof}

\begin{corollary}
For ordered $\mathbb F_2$-bases $B,B'$ of $\mathbb F_{2^8}$,
\[
S_{{\rm AES},B}=S_{{\rm AES},B'}
\iff
B=B'.
\]
\end{corollary}
\begin{proof}
By the basis-equivalence criterion,
\[
S_{{\rm AES},B}=S_{{\rm AES},B'}
\iff
P_{B'}P_B^{-1}\in\operatorname{Stab}(S_{\rm AES}).
\]
Since $\operatorname{Stab}(S_{\rm AES})=\{I\}$, we obtain
$P_{B'}=P_B$, hence $B=B'$. The converse is immediate.
\end{proof}

The AES result shows that the AES S-box linear layer completely
eliminates the Frobenius redundancy of inversion. This motivates the
broader question addressed in the next section. How often does such
rigidity occur when the outer linear transformation varies?

\section{Generic Rigidity}
\label{sec:generic-rigidity}
The AES result shows that a fixed outer linear transformation can make
transformed inversion basis rigid. In this section, we show that such
rigidity is generic when the outer invertible linear transformation
varies. Our argument reduces the existence of a nontrivial linear
symmetry to a conjugacy problem for the semilinear candidates arising
from two sided linear equivalences of inversion. This leads to an
asymptotic rigidity theorem and sharper finite dimensional bounds.

\subsection{Stabilizer Events as Conjugacy Conditions}
\label{stabilizer_as_conjugacy}

Recall the notation of Section~2, and let
\[
G=\mathrm{GL}_n(\mathbb F_2),
\qquad
F_A:=F_{A,0}=A\circ J_n,
\qquad A\in G.
\]
For the semilinear candidates $T_{a,r}=M_a\sigma^r$, define
\[
U_{a,r}:=M_{a^{-1}}\sigma^r.
\]

By Theorem~\ref{thm:stabilizer-reduction}, applied with $c=0$, we have
\[
\operatorname{Stab}(F_A)
=
\left\{
T_{a,r}:
A^{-1}T_{a,r}A=U_{a,r}
\right\}.
\]
Since $T_{1,0}=I$, the stabilizer is nontrivial precisely when
\[
A^{-1}T_{a,r}A=U_{a,r}
\]
for some $(a,r)\neq(1,0)$.

For $T,U\in G$, define the conjugator set
\[
\mathcal C(T,U)
=
\{A\in G:A^{-1}TA=U\},
\]
and write
\[
C_G(T)
=
\{C\in G:CT=TC\}
\]
for the group centralizer of $T$ in $G$.

\begin{lemma}
\label{lem:conjugator-cosets}
Let $T,U\in G$. If $T$ and $U$ are not conjugate in $G$, then
\[
\mathcal C(T,U)=\varnothing.
\]
Otherwise, for any $A_0\in\mathcal C(T,U)$,
\[
\mathcal C(T,U)
=
A_0C_G(U)
=
C_G(T)A_0.
\]
In particular,
\[
|\mathcal C(T,U)|
=
|C_G(T)|
=
|C_G(U)|.
\]
\end{lemma}

\begin{proof}
Suppose $A_0^{-1}TA_0=U$ and write $A=A_0C$. Then
\[
A^{-1}TA=U
\iff
C^{-1}UC=U,
\]
so
\[
\mathcal C(T,U)=A_0C_G(U).
\]
The second coset description follows similarly, and conjugation by
$A_0$ identifies $C_G(T)$ with $C_G(U)$.
\end{proof}

Consequently, if $A$ is chosen uniformly from $G$, then
\[
\Pr_A(A^{-1}TA=U)
=
\begin{cases}
\dfrac{|C_G(T)|}{|G|}, & T\sim U,\\[1ex]
0, & T\not\sim U,
\end{cases}
\]
where $\sim$ denotes conjugacy in $G$.

Let
\[
p_n
=
\Pr_A\!\left(
\operatorname{Stab}(F_A)\neq\{I\}
\right).
\]
Applying the union bound over all nonidentity semilinear candidates gives
\begin{equation}
\label{equpn}
p_n
\le
\frac{1}{|G|}
\sum_{\substack{a\in\mathbb F^\times,\ 0\le r<n\\
(a,r)\neq(1,0)}}
\mathbf 1_{\{T_{a,r}\sim U_{a,r}\}}
\,|C_G(T_{a,r})|.
\end{equation}

Thus the generic rigidity problem reduces to determining when
$T_{a,r}$ and $U_{a,r}$ are conjugate and controlling the size of their
centralizers. We first address the conjugacy condition.

\subsection{Conjugacy of Semilinear Candidates}
\label{subsec:semilinear-conjugacy}

Equation~\eqref{equpn} shows that a semilinear candidate contributes to the
exceptional probability only when $T_{a,r}$ and $U_{a,r}$ are conjugate
in $G$. We now characterize exactly when this occurs. We use standard facts about finite fields, Frobenius automorphisms, and
relative norms~\cite{lidl1997finite}.

For $0\le r<n$, define
\[
d_r=\gcd(n,r),
\qquad
K_r=\mathbb F_{2^{d_r}},
\]
where $\gcd(n,0)=n$, and let
\[
z_{a,r}
=
N_{\mathbb F/K_r}(a)
\]
denote the norm of $a$ from the extension
$\mathbb F/K_r$.
For $r=0$, the norm is the identity, so $z_{a,0}=a$. We first have two elementary observations.

\begin{lemma}
\label{lem:semilinear-power}
For every $a\in\mathbb F^\times$ and $0\le r<n$,
\[
T_{a,r}^{\,n/d_r}=M_{z_{a,r}},
\qquad
U_{a,r}^{\,n/d_r}=M_{z_{a,r}^{-1}}.
\]
\end{lemma}

\begin{proof}
Let $m=n/d_r$. Since $\sigma^r$ generates
$\operatorname{Gal}(\mathbb F/K_r)$,
\[
T_{a,r}^{\,m}
=
M_{a\,a^{2^r}\cdots a^{2^{(m-1)r}}}
=
M_{N_{\mathbb F/K_r}(a)}
=
M_{z_{a,r}}.
\]
The second identity follows by replacing $a$ with $a^{-1}$.
\end{proof}

\begin{lemma}
\label{lem:multiplication-conjugacy}
Let $z,w\in K_r^\times$. Then the multiplication maps $M_z$ and
$M_w$ are conjugate in $\mathrm{GL}_n(\mathbb F_2)$ if and only if
\[
w=z^{2^k}
\]
for some $0\le k<d_r$.
\end{lemma}

\begin{proof}
Let
\[
e_z=[\mathbb F_2(z):\mathbb F_2].
\]
As an $\mathbb F_2$-linear transformation, $M_z$ has minimal polynomial
equal to the minimal polynomial of $z$ over $\mathbb F_2$, with each
corresponding rational canonical block occurring $n/e_z$ times.
Consequently, $M_z$ and $M_w$ are conjugate if and only if $z$ and $w$
have the same minimal polynomial over $\mathbb F_2$.

Over a finite field, this is equivalent to $z$ and $w$ belonging to the
same Frobenius orbit. Since $z,w\in K_r=\mathbb F_{2^{d_r}}$, this means
\[
w=z^{2^k}
\]
for some $0\le k<d_r$.
\end{proof}

We can now determine the conjugacy condition appearing in~\eqref{equpn}.

\begin{theorem}
\label{thm:semilinear-conjugacy}
Let $a\in\mathbb F^\times$ and $0\le r<n$. Then
\[
T_{a,r}\sim U_{a,r}
\iff
z_{a,r}^{-1}=z_{a,r}^{2^k}
\text{ for some }0\le k<d_r
\iff
z_{a,r}^{2^k+1}=1
\text{ for some }0\le k<d_r.
\]
\end{theorem}

\begin{proof}
Suppose first that
\[
T_{a,r}\sim U_{a,r}.
\]
Taking the $(n/d_r)$th power preserves conjugacy, so
Lemma~\ref{lem:semilinear-power} gives
\[
M_{z_{a,r}}\sim M_{z_{a,r}^{-1}}.
\]
Lemma~\ref{lem:multiplication-conjugacy} therefore implies
\[
z_{a,r}^{-1}=z_{a,r}^{2^k}
\]
for some $0\le k<d_r$.

Conversely, suppose
\[
z_{a,r}^{-1}=z_{a,r}^{2^k}
\]
for some $0\le k<d_r$.

If $r=0$, then $z_{a,0}=a$, and
\[
a^{-1}=a^{2^k}.
\]
Since
\[
\sigma^k M_a\sigma^{-k}=M_{a^{2^k}},
\]
we immediately obtain
\[
M_a\sim M_{a^{-1}}.
\]

Now assume $1\le r<n$. Since Frobenius powers commute,
\[
\sigma^k T_{a,r}\sigma^{-k}
=
M_{a^{2^k}}\sigma^r.
\]
Moreover,
\[
N_{\mathbb F/K_r}(a^{2^k})
=
z_{a,r}^{2^k}
=
z_{a,r}^{-1}
=
N_{\mathbb F/K_r}(a^{-1}).
\]
Hence
\[
N_{\mathbb F/K_r}
\left(
\frac{a^{-1}}{a^{2^k}}
\right)
=1.
\]
By the finite field form of Hilbert's Theorem~90, there exists
$b\in\mathbb F^\times$ such that
\[
b^{2^r-1}
=
\frac{a^{-1}}{a^{2^k}}.
\]
Therefore
\[
M_b^{-1}
\bigl(M_{a^{2^k}}\sigma^r\bigr)
M_b
=
M_{a^{2^k}b^{2^r-1}}\sigma^r
=
M_{a^{-1}}\sigma^r
=
U_{a,r}.
\]
Thus $T_{a,r}$ and $U_{a,r}$ are conjugate.
\end{proof}

The theorem gives an explicit criterion for deciding which semilinear
candidates can contribute to the exceptional set. In particular, the
conjugacy indicator in~\eqref{equpn} can now be written entirely in terms of the
relative norm $z_{a,r}$.

\begin{corollary}
\label{cor:conjugacy-indicator}
The conjugacy indicator in~\eqref{equpn} is given explicitly by
\[
\mathbf 1_{\{T_{a,r}\sim U_{a,r}\}}
=
\mathbf 1_{\{
z_{a,r}^{-1}
=
z_{a,r}^{2^k}
\text{ for some }0\le k<d_r
\}}.
\]
Equivalently, a candidate contributes only if
\[
\operatorname{ord}(z_{a,r})
\mid 2^k+1
\]
for some $0\le k<d_r$.
\end{corollary}

Theorem~\ref{thm:semilinear-conjugacy} makes the conjugacy
condition in~\eqref{equpn} explicit in terms of the relative
norm \(z_{a,r}\). It characterizes exactly which semilinear
candidates can stabilize \(A\circ J_n\) for some invertible
outer linear map \(A\). 
To bound the contributions of the admissible candidates,
we next estimate the sizes of their group centralizers.

\subsection{Centralizers of Semilinear Maps}
\label{centralizers_semilinear_maps}
The conjugacy criterion determines which semilinear candidates can
contribute to the exceptional set. To count the outer linear maps
that realize each admissible conjugacy, we now study their centralizers.
For classical work on centralizers of semilinear transformations, see~\cite{smith1976centralizer}.
Here we require an explicit dimension formula for the specific
semilinear maps
\[
T_{a,r}=M_a\sigma^r
\]
arising from inversion. The maps $T_{a,r}(x)=a x^{2^r}$ are semilinear with respect to the
Frobenius automorphism \(\sigma^r\) of $\mathbb F$, while they are linear over
$\mathbb F_2$.

\begin{definition}
\label{def:centralizer-algebra}
For $T\in G$, define
\[
Z(T)
=
\{L\in\operatorname{End}_{\mathbb F_2}(\mathbb F):LT=TL\}.
\]
We call $Z(T)$ the \emph{centralizer algebra} of $T$.
Its invertible elements form the group centralizer
\[
C_G(T)=Z(T)\cap G.
\]
\end{definition}

Thus
\[
|C_G(T)|
\le
|Z(T)|
=
2^{\dim_{\mathbb F_2}Z(T)}.
\]
We next determine the dimension of $Z(T_{a,r})$.

\begin{lemma}
\label{lem:multiplication-centralizer}
Let $a\in\mathbb F^\times$, and put
\[
e_a=[\mathbb F_2(a):\mathbb F_2].
\]
Then
\[
Z(M_a)
=
\operatorname{End}_{\mathbb F_{2^{e_a}}}(\mathbb F),
\]
and hence
\[
\dim_{\mathbb F_2}Z(M_a)
=
\frac{n^2}{e_a}.
\]
\end{lemma}

\begin{proof}
An $\mathbb F_2$-linear map $L$ commutes with $M_a$ if and only if it
commutes with multiplication by every element of
\[
\mathbb F_2[a]=\mathbb F_2(a)=\mathbb F_{2^{e_a}}.
\]
Equivalently, $L$ is $\mathbb F_{2^{e_a}}$-linear. Hence
\[
Z(M_a)
=
\operatorname{End}_{\mathbb F_{2^{e_a}}}(\mathbb F).
\]

Since $\mathbb F$ has dimension $n/{e_a}$ over $\mathbb F_{2^{e_a}}$,
the endomorphism algebra has dimension $(n/{e_a})^2$ over
$\mathbb F_{2^{e_a}}$. Therefore its dimension over $\mathbb F_2$ is
\[
e_a\left(\frac{n}{e_a}\right)^2
=
\frac{n^2}{e_a}.
\]
\end{proof}

The case $r\neq0$ is controlled by the relative norm introduced in
subsection~\ref{subsec:semilinear-conjugacy}.

\begin{lemma}
\label{lem:twisted-centralizer}
Let $a\in\mathbb F^\times$ and $1\le r<n$. Write
\[
d_r=\gcd(n,r),
\qquad
z_{a,r}=N_{\mathbb F/\mathbb F_{2^{d_r}}}(a),
\]
and let
\[
e_{a,r}
=
[\mathbb F_2(z_{a,r}):\mathbb F_2].
\]
Then
\[
\dim_{\mathbb F_2} Z(T_{a,r})
=
\frac{n d_r}{e_{a,r}}.
\]
\end{lemma}

\begin{proof}
Every $\mathbb F_2$-linear map $L:\mathbb F\to\mathbb F$ has a unique
linearized polynomial representation
\[
L(x)
=
\sum_{i=0}^{n-1}\ell_i x^{2^i},
\qquad
\ell_i\in\mathbb F.
\]
Then
\[
L T_{a,r}(x)
=
\sum_{i=0}^{n-1}
\ell_i a^{2^i}x^{2^{i+r}},
\]
whereas
\[
T_{a,r}L(x)
=
\sum_{i=0}^{n-1}
a\ell_i^{2^r}x^{2^{i+r}}.
\]
Since the Frobenius exponents are taken modulo $n$, addition by $r$
permutes the indices. Uniqueness of the linearized polynomial
representation therefore gives
\begin{equation}
\label{equ5}
L T_{a,r}=T_{a,r}L
\iff
\ell_i^{2^r}
=
a^{2^i-1}\ell_i
\qquad
(0\le i<n).
\end{equation}

Consider more generally
\[
u^{2^r}=\eta u,
\qquad
\eta\in\mathbb F^\times.
\]
If a nonzero solution $u_0$ exists, then every nonzero solution has the
form
\[
u_0v,
\qquad
v\in\mathbb F_{2^{d_r}}^\times,
\]
because
\[
\left(\frac{u}{u_0}\right)^{2^r}
=
\frac{u}{u_0}.
\]
Hence the solution space has dimension $d_r$ over $\mathbb F_2$.

By the finite field form of Hilbert's Theorem~90, a nonzero solution
exists if and only if
\[
N_{\mathbb F/\mathbb F_{2^{d_r}}}(\eta)=1.
\]
Applying this to the $i$th equation in~\eqref{equ5}, with
\[
\eta=a^{2^i-1},
\]
gives
\[
N_{\mathbb F/\mathbb F_{2^{d_r}}}
\left(a^{2^i-1}\right)
=
z_{a,r}^{\,2^i-1}
=
1.
\]
Equivalently,
\[
z_{a,r}^{2^i}=z_{a,r}.
\]

Since
\[
e_{a,r}
=
[\mathbb F_2(z_{a,r}):\mathbb F_2],
\]
this holds exactly when
\[
e_{a,r}\mid i.
\]
Moreover,
\[
z_{a,r}\in\mathbb F_{2^{d_r}},
\]
so
\[
e_{a,r}\mid d_r\mid n.
\]
Therefore exactly
\[
\frac{n}{e_{a,r}}
\]
indices in $\{0,\ldots,n-1\}$ satisfy this condition.

Each admissible coefficient $\ell_i$ contributes a
$d_r$-dimensional solution space, and the coefficient equations are
independent. Hence
\[
\dim_{\mathbb F_2} Z(T_{a,r})
=
\frac{n}{e_{a,r}}\,d_r
=
\frac{n d_r}{e_{a,r}}.
\]
\end{proof}

For $r=0$, we use the convention
\[
d_0=n,
\qquad
z_{a,0}=a,
\qquad
e_{a,0}=[\mathbb F_2(a):\mathbb F_2].
\]
Lemmas~\ref{lem:multiplication-centralizer} and
\ref{lem:twisted-centralizer} therefore give the unified formula
\[
\label{eq6}
\dim_{\mathbb F_2}Z(T_{a,r})
=
\frac{n d_r}{e_{a,r}}
\]

for every $a\in\mathbb F^\times$ and $0\le r<n$.

\subsection{Probability of Nontrivial Stabilizers}
\label{subsec:generic-rigidity}

We now combine the conjugacy criterion of subsection~\ref{subsec:semilinear-conjugacy} with the
centralizer formulas of subsection~\ref{centralizers_semilinear_maps}. Let $A$ be chosen uniformly from
\[
G=\mathrm{GL}_n(\mathbb F_2),
\]
and define
\[
p_n
=
\Pr_A\!\left(
\operatorname{Stab}(A\circ J_n)\neq\{I\}
\right).
\]

Recall from subsection~\ref{stabilizer_as_conjugacy},
\[
p_n
\le
\frac{1}{|G|}
\sum_{\substack{a\in\mathbb F^\times,\ 0\le r<n\\
(a,r)\neq(1,0)}}
\mathbf 1_{\{T_{a,r}\sim U_{a,r}\}}
\,|C_G(T_{a,r})|.
\]
The conjugacy criterion of subsection~\ref{subsec:semilinear-conjugacy} gives
\[
T_{a,r}\sim U_{a,r}
\iff
z_{a,r}^{2^k+1}=1
\quad\text{for some }0\le k<d_r,
\]
while subsection~\ref{centralizers_semilinear_maps} gives
\[
|C_G(T_{a,r})|
\le
2^{n d_r/e_{a,r}}.
\]
Hence
\[
p_n
\le
\frac{1}{|G|}
\sum_{\substack{a\in\mathbb F^\times,\ 0\le r<n\\
(a,r)\neq(1,0)}}
\mathbf 1_{\{
z_{a,r}^{2^k+1}=1
\text{ for some }0\le k<d_r
\}}
\,2^{n d_r/e_{a,r}}.
\]

The preceding structural results now yield the main probabilistic
rigidity theorem of this section. 

\begin{theorem}
\label{thm:generic-rigidity}
For $n\ge4$,
\[
p_n
\le
\frac{n}{c_0}\,2^{-n^2/2+n},
\]
where
\[
c_0
=
\prod_{j=1}^{\infty}(1-2^{-j})>0.
\]
In particular,
\[
p_n\le 2^{-\Omega(n^2)}.
\]
\end{theorem}

\begin{proof}
We first obtain a uniform bound on the centralizers of all nonidentity
semilinear candidates. Recall from subsection~\ref{centralizers_semilinear_maps} that
\[
\dim_{\mathbb F_2} Z(T_{a,r})
=
\frac{n d_r}{e_{a,r}}.
\]

Suppose first that $r=0$. Then $d_0=n$ and
\[
e_{a,0}=[\mathbb F_2(a):\mathbb F_2].
\]
Since $(a,0)\neq(1,0)$, we have $a\neq1$. As
$\mathbb F_2^\times=\{1\}$, this implies $e_{a,0}\ge2$. Hence
\[
\dim_{\mathbb F_2} Z(T_{a,0})
=
\frac{n^2}{e_{a,0}}
\le
\frac{n^2}{2}.
\]

Now suppose that $1\le r<n$. Then $d_r=\gcd(n,r)$ is a proper divisor
of $n$, and therefore
\[
d_r\le\frac n2.
\]
Since $e_{a,r}\ge1$, we obtain
\[
\dim_{\mathbb F_2} Z(T_{a,r})
=
\frac{n d_r}{e_{a,r}}
\le
n d_r
\le
\frac{n^2}{2}.
\]
Thus every nonidentity candidate satisfies
\begin{equation}
\label{eq8}
|C_G(T_{a,r})|
\le
|Z(T_{a,r})|
\le
2^{n^2/2}.
\end{equation}

There are
\[
n(2^n-1)-1<n2^n
\]
nonidentity pairs $(a,r)$. Dropping the conjugacy indicator from the
union bound in subsection~\ref{stabilizer_as_conjugacy} and applying~\eqref{eq8} therefore gives
\[
p_n
\le
\frac{n(2^n-1)-1}{|G|}\,2^{n^2/2}
<
\frac{n2^n}{|G|}\,2^{n^2/2}.
\]

Finally,
\[
|G|
=
|\mathrm{GL}_n(\mathbb F_2)|
=
2^{n^2}\prod_{j=1}^n(1-2^{-j}).
\]
Since
\[
\prod_{j=1}^n(1-2^{-j})
\ge
c_0,
\qquad
c_0=\prod_{j=1}^{\infty}(1-2^{-j})>0,
\]
we conclude that
\[
p_n
<
\frac{n2^n\,2^{n^2/2}}
{c_0\,2^{n^2}}
=
\frac{n}{c_0}\,2^{-n^2/2+n}.
\]
In particular, $p_n$ is bounded by $2^{-\Omega(n^2)}$.
\end{proof}

\begin{remark}
For a fixed $c\in\mathbb F_{2^n}$, Theorem~\ref{thm:stabilizer-reduction}
gives
\[
\operatorname{Stab}(F_{A,c})
=
\{T\in\operatorname{Stab}(A\circ J_n):T(c)=c\}
\subseteq
\operatorname{Stab}(A\circ J_n).
\]
Consequently, if $A$ is chosen uniformly from
$\mathrm{GL}_n(\mathbb F_2)$, then
\[
\Pr_A\!\left(
\operatorname{Stab}(F_{A,c})\ne\{I\}
\right)
\le p_n
\le
\frac{n}{c_0}2^{-n^2/2+n},
\qquad n\ge4.
\]
The refined upper bounds for $p_n$ derived below therefore also apply
to $F_{A,c}$.
\end{remark}

\subsection{Refined Finite Dimensional Bounds}
\label{subsec:refined-bounds}

The uniform estimate in the preceding subsection discards the exact
conjugacy condition. Retaining this condition together with the exact
centralizer dimensions gives the sharper bound
\[
B_n^{\mathrm{conj}}
=
\frac{1}{|G|}
\sum_{\substack{a\in\mathbb F^\times,\ 0\le r<n\\
(a,r)\neq(1,0)}}
\mathbf 1_{\{
z_{a,r}^{2^k+1}=1
\text{ for some }0\le k<d_r
\}}
\,2^{n d_r/e_{a,r}}.
\]
By subsection~\ref{subsec:generic-rigidity},
\[
p_n\le B_n^{\mathrm{conj}}.
\]
For comparison, omitting the conjugacy condition gives
\[
B_n
=
\frac{1}{|G|}
\sum_{\substack{a\in\mathbb F^\times,\ 0\le r<n\\
(a,r)\neq(1,0)}}
2^{n d_r/e_{a,r}},
\]
and therefore
\[
p_n
\le
B_n^{\mathrm{conj}}
\le
B_n.
\]

To evaluate the refined bound more explicitly, we first characterize
which norm values can satisfy the conjugacy condition.

\begin{lemma}
\label{lem:inverse-frobenius-degree}
Let $z\in\mathbb F_{2^d}^\times$, and let
\[
e_z=[\mathbb F_2(z):\mathbb F_2].
\]
If $e_z>1$, then
\[
z^{-1}=z^{2^k}
\]
for some integer \(0\le k<e_z\) if and only if $e_z$ is even and
\[
z^{-1}=z^{2^{e_z/2}}.
\]
Equivalently,
\[
z^{2^{e_z/2}+1}=1.
\]
\end{lemma}

\begin{proof}
Suppose that
\[
z^{-1}=z^{2^k}.
\]
Applying the same Frobenius power again gives
\[
z^{2^{2k}}=z.
\]
Since the Frobenius orbit of $z$ has length $e_z$, it follows that
\[
e_z\mid 2k.
\]
On the other hand, $e_z\nmid k$, since otherwise
$z^{2^k}=z$ and hence $z=z^{-1}$, which forces $z=1$ in
characteristic two, contradicting $e_z>1$.

Thus $e_z$ is even and
\[
k\equiv \frac{e_z}{2}\pmod{e_z}.
\]
Hence
\[
z^{-1}=z^{2^{e_z/2}}.
\]
The converse is immediate, and the last equality is equivalent to
\[
z^{2^{e_z/2}+1}=1.
\]
\end{proof}

For \(e\ge1\), let \(Q_e\) denote the number of nonzero elements \(z\in\mathbb F_{2^e}\) of exact degree \(e\) over \(\mathbb F_2\) whose multiplicative inverse belongs to the same Frobenius orbit as \(z\).
Then
\[
Q_1=1,
\]
and Lemma~\ref{lem:inverse-frobenius-degree} gives
\[
Q_e=0
\qquad\text{for odd }e>1.
\]
For even $e$,
\[
Q_e
=
\#\left\{
z\in\mathbb F_{2^e}^\times:
[\mathbb F_2(z):\mathbb F_2]=e,\;
z^{2^{e/2}+1}=1
\right\}.
\]

We can now group the refined bound according to the degree of the norm
value.

\begin{proposition}
\label{prop:grouped-refined-bound}
For every $n$,
\[
B_n^{\mathrm{conj}}
=
\frac{1}{|G|}
\left[
\sum_{r=0}^{n-1}
\frac{2^n-1}{2^{d_r}-1}
\sum_{e\mid d_r}
Q_e\,2^{n d_r/e}
-
2^{n^2}
\right].
\]
\end{proposition}

\begin{proof}
Fix $r$. The norm map
\[
N_{\mathbb F/\mathbb F_{2^{d_r}}}:
\mathbb F^\times
\longrightarrow
\mathbb F_{2^{d_r}}^\times
\]
is surjective, and each norm value has
\[
\frac{2^n-1}{2^{d_r}-1}
\]
preimages.

If a norm value $z$ has exact degree $e$ over $\mathbb F_2$, then
$e\mid d_r$ and the corresponding centralizer contribution is
\[
2^{n d_r/e}.
\]
Among the elements of exact degree $e$, precisely $Q_e$ satisfy the
conjugacy condition. Hence the total contribution for fixed $r$ is
\[
\frac{2^n-1}{2^{d_r}-1}
\sum_{e\mid d_r}
Q_e\,2^{n d_r/e}.
\]
Summing over $r$ also includes the identity candidate $(a,r)=(1,0)$,
whose contribution is $2^{n^2}$. Removing this term gives the stated
formula.
\end{proof}

This grouped expression gives a computable finite dimensional refinement
of the generic rigidity bound.

\section{Computational Verification}
\label{sec:computational-verification}

We complement the theoretical results with three computations.
First, we independently determine the linear stabilizers of
the AES S-box and two related maps from their complete truth
tables. Second, we enumerate all invertible linear maps in small dimensions to obtain exact probabilities of
nontrivial stabilizers. Third, we sample uniformly random
invertible linear maps in larger dimensions and compare the
observed frequencies with the theoretical bounds.

The computations were performed in Ubuntu 22.04 using
Python 3.10.12. Binary vectors are represented by integers,
with vector addition implemented by bitwise XOR. Linear maps
are represented by their images on the standard basis.
The computational routines use only the Python standard
library, plotting uses Matplotlib.

For finite field arithmetic, we use
\[
\mathbb F
=
\mathbb F_2[x]/(m_n(x)),
\]
with the following irreducible polynomials:
\[
\begin{array}{c|l}
n & m_n(x)\\
\hline
3 & x^3+x+1\\
4 & x^4+x+1\\
5 & x^5+x^2+1\\
6 & x^6+x+1\\
7 & x^7+x+1\\
8 & x^8+x^4+x^3+x+1.
\end{array}
\]
For $n=8$, this is the standard AES field representation.
All nonzero inverses are checked by multiplication. 

All computations were performed on a single CPU thread on an HP desktop
computer with an Intel Core i7-13700F processor and 32~GB of RAM.

\subsection{Independent Verification of AES Rigidity}

We search for all invertible $\mathbb F_2$-linear maps $T$
satisfying
\[
T\circ F=F\circ T
\]
directly from the complete truth table of $F$.
This search does not use the semilinear characterization
of stabilizers or restrict the candidates to maps
of the form $M_a\sigma^r$.

The algorithm maintains an injective linear map on a
subspace of $\mathbb F_2^n$. Assigning $T(x)=y$ for a
vector outside the current domain extends the map to
the enlarged span by linearity. Every assigned pair
$(x,y)$ also forces
\[
T(F(x))=F(y).
\]
These consequences are propagated until no further
assignments are forced or a conflict occurs.
Branches with inconsistent assignments or a failure
of injectivity are discarded.

If the domain remains a proper subspace, the algorithm
chooses the first unassigned standard basis vector
and considers every image outside the current image
subspace. Consequently, every invertible linear map
satisfying the commuting relation occurs in the search.
Each completed map is additionally checked on the
entire truth table.

The initial assignment $T(0)=0$ forces
$T(F(0))=F(0)$. For the full AES S-box, closure from
this constraint already determines the identity map.
For the other two maps, branching is required.

We apply this procedure to the literal 256-entry AES
S-box table, the pure inversion map $J_8$, and
$A_{\mathrm{AES}}\circ J_8$. The AES table is also
checked against its field-inversion and affine
representation. The results are shown in
Table~\ref{tab:aes-exact-verification}.

\begin{table}[t]
\centering
\caption{Independent exhaustive closure search.
Times measure the stabilizer search for each map.}
\label{tab:aes-exact-verification}
\begin{tabular}{lrr}
\hline
Map & Stabilizer order & Time (s)\\
\hline
$S_{\mathrm{AES}}$              & 1 & 0.00031\\
$J_8$                           & 8 & 0.00402\\
$A_{\mathrm{AES}}\circ J_8$      & 1 & 0.00250\\
\hline
\end{tabular}
\end{table}

For both AES maps, the unique solution is the identity.
For $J_8$, the eight recovered matrices are compared
individually with
\[
I,\sigma,\ldots,\sigma^7,
\]
and the two sets agree exactly. Thus the pure inversion case provides a nontrivial control for the
search, while the AES computations independently confirm the rigidity
results of subsection~\ref{subsec:aes-basis-rigidity}.

\subsection{Exact Probabilities and Refined Bounds}
\label{subsec:exact-probabilities}

We compute the exact exceptional probabilities in small dimensions
and evaluate the bounds obtained in Section~3. We also check the
conjugacy criterion and the centralizer dimension formula directly
from the binary matrices of the semilinear candidates.

For each $n\in\{3,4,5\}$, we enumerate all
$A\in\mathrm{GL}_n(\mathbb F_2)$ by constructing ordered linearly
independent column tuples. Each invertible matrix appears exactly
once. We count
\[
N_n=\#\left\{
A\in\mathrm{GL}_n(\mathbb F_2):
\operatorname{Stab}(A\circ J_n)\ne\{I\}
\right\},
\qquad
p_n=\frac{N_n}{|\mathrm{GL}_n(\mathbb F_2)|}.
\]
For $n=3,4$, every map $A\circ J_n$ is tested by the independent
closure search of the preceding subsection. The resulting decision
is also compared with the semilinear candidate test for every $A$,
with no disagreements. In particular, the exact computation for
$n=3$ does not require extending the inversion-equivalence theorem
beyond its stated range.

For $n=5$, we use the stabilizer characterization and test
\[
T_{a,r}A=AU_{a,r}
\]
for nonidentity pairs $(a,r)$. Since both sides are linear, equality
is checked on the standard basis. The enumeration covers all
$9{,}999{,}360$ invertible matrices and takes approximately
$96$ seconds. The exact results are given in
Table~\ref{tab:exact-probabilities}.

\begin{table}[t]
\centering
\caption{Exact probabilities of nontrivial stabilizers.
Times for $n=3,4$ include both closure search and the semilinear test.}
\label{tab:exact-probabilities}
\begin{tabular}{rrrrr}
\hline
$n$ & $|\mathrm{GL}_n(\mathbb F_2)|$ & $N_n$ & $p_n$ & Time (s)\\
\hline
3 & 168 & 21 & $1/8$ & 0.005\\
4 & 20,160 & 600 & $5/168$ & 1.35\\
5 & 9,999,360 & 465 & $1/21504$ & 96.38\\
\hline
\end{tabular}
\end{table}

We next verify the structural formulas for every nonidentity
semilinear candidate in dimensions $n=3,\ldots,8$.
For each pair $(a,r)$, we construct the binary matrices of
$T_{a,r}$ and $U_{a,r}$ and determine whether they are conjugate
by comparing their elementary divisors. These are recovered from
the characteristic polynomials and the nullities of powers of
their irreducible factors evaluated at the matrices.
This direct test does not use the norm criterion.

For all $3{,}537$ candidates, the direct conjugacy test agrees with
\[
z_{a,r}^{2^k+1}=1
\qquad\text{for some }0\le k<d_r.
\]
We compute the nullity of the binary linear system independently
\[
LT_{a,r}-T_{a,r}L=0.
\]
Every computed dimension agrees with $nd_r/e_{a,r}$.
Finally, evaluating $B_n^{\mathrm{conj}}$ by its candidate sum
and by the grouped expression in
Proposition~\ref{prop:grouped-refined-bound}
gives exactly the same rational value in each dimension.

For comparison, write
\[
b_n=\frac{n}{c_0}2^{-n^2/2+n}
\]
for the uniform bound, and retain the notation
$B_n$ and $B_n^{\mathrm{conj}}$ from
subsection~\ref{subsec:refined-bounds}.
Table~\ref{tab:bound-comparison} reports these bounds for
$n=4,\ldots,8$, along with the available exact probabilities.
The exact $n=3$ result is reported separately because the
probability theorem is stated for $n\ge4$.

\begin{table}[t]
\centering
\caption{Theoretical upper bounds and exact probabilities.
A dash indicates that the exact probability was not computed.}
\label{tab:bound-comparison}
\begin{tabular}{rrrrr}
\hline
$n$ & $b_n$ & $B_n$ & $B_n^{\mathrm{conj}}$ & Exact $p_n$\\
\hline
4 & $8.65687\times10^{-1}$ & $1.30159\times10^{-1}$
  & $1.23810\times10^{-1}$ & $2.97619\times10^{-2}$\\
5 & $9.56458\times10^{-2}$ & $4.92832\times10^{-4}$
  & $3.96825\times10^{-4}$ & $4.65030\times10^{-5}$\\
6 & $5.07238\times10^{-3}$ & $1.53807\times10^{-4}$
  & $1.52264\times10^{-4}$ & ---\\
7 & $1.30766\times10^{-4}$ & $6.93708\times10^{-10}$
  & $5.95276\times10^{-10}$ & ---\\
8 & $1.65117\times10^{-6}$ & $1.52614\times10^{-8}$
  & $1.52613\times10^{-8}$ & ---\\
\hline
\end{tabular}
\end{table}

The conjugacy condition improves the bound in every displayed
dimension, but the improvement varies with $n$. For example, relative to $B_n$, the reduction is approximately $19.48\%$
for $n=5$ and $14.19\%$ for $n=7$, whereas it is only
$0.000755\%$ for $n=8$.

For $n=4,5$, the exact probabilities lie strictly below
$B_n^{\mathrm{conj}}$. The remaining gap reflects two sources
of overcounting: the replacement of group centralizers by their
full centralizer algebras, and the overlap among the stabilizer
events in the union bound.
The finite dimensional bounds $B_n$ and $B_n^{\mathrm{conj}}$
depend on the divisor structure of $n$ and on the degrees of the
norm values, so they need not decrease monotonically with $n$.
Figure~\ref{fig:bound-comparison} illustrates this behavior on a
logarithmic scale.

\begin{figure}[t]
\centering
\includegraphics[width=0.6\textwidth]{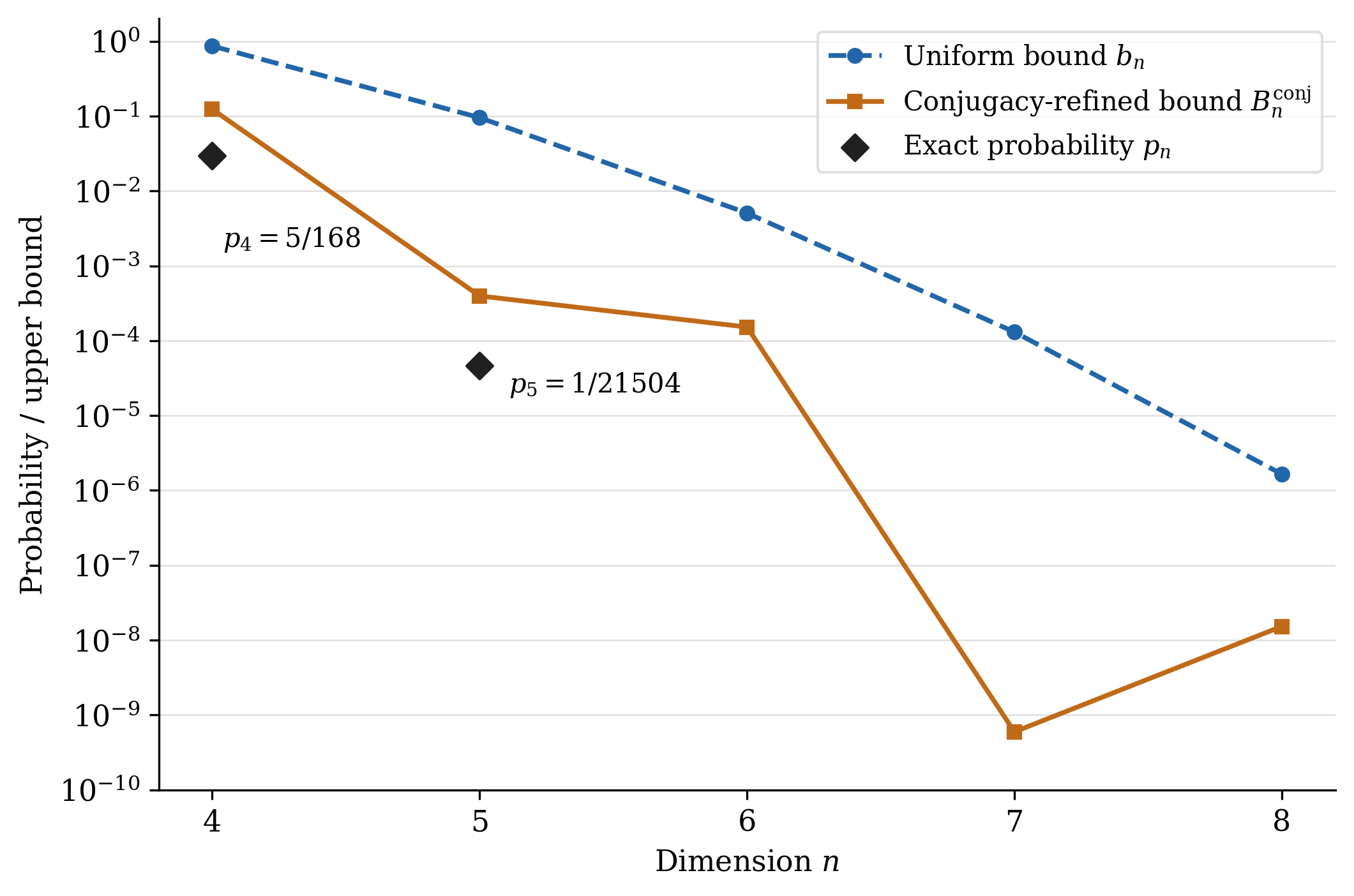}
\caption{Uniform and conjugacy-refined upper bounds for the exceptional
probability with the exact values of $p_n$ for $n=4,5$ on a
logarithmic scale.}
\label{fig:bound-comparison}
\end{figure}

\subsection{Monte Carlo Estimates}
\label{subsec:monte-carlo}

For each $n\in\{5,6,7,8\}$, we independently sample
$100{,}000$ matrices from $\mathrm{GL}_n(\mathbb F_2)$.
Each successive column is chosen uniformly outside the span
of the preceding columns. Thus every ordered basis has probability
\[
\prod_{k=0}^{n-1}(2^n-2^k)^{-1}
=
|\mathrm{GL}_n(\mathbb F_2)|^{-1},
\]
so the resulting matrices are uniformly distributed.
Sampling is with replacement across trials, and the pseudorandom
generator uses seed $20260906+n$ in dimension $n$.

Each sampled matrix $A$ is tested for a nontrivial stabilizer
using the condition
\[
T_{a,r}A=AU_{a,r},
\qquad (a,r)\ne(1,0).
\]
As an additional consistency check, we compare the semilinear
test with independent closure search on random matrices and
prescribed examples with nontrivial stabilizers, with no
disagreements.

Let $m=100{,}000$, let $k_n$ denote the number of samples with
a nontrivial stabilizer, and put
\[
\widehat p_n=\frac{k_n}{m}.
\]
Table~\ref{tab:monte-carlo} reports the observed frequencies
and pointwise $95\%$ Wilson score confidence intervals~\cite{wilson1927probable}.
With $z=1.959964$, the interval endpoints are
\[
\frac{
\widehat p_n+\dfrac{z^2}{2m}
\;\pm\;
z\sqrt{
\dfrac{\widehat p_n(1-\widehat p_n)}{m}
+\dfrac{z^2}{4m^2}
}
}{
1+\dfrac{z^2}{m}
}.
\]
These intervals are computed separately for each dimension
and are not simultaneous confidence bands.

\begin{table}[t]
\centering
\caption{Monte Carlo estimates from $100{,}000$ samples per
dimension, with pointwise $95\%$ Wilson score intervals.}
\label{tab:monte-carlo}
\begin{tabular}{rrrrr}
\hline
$n$ & $k_n$ & $\widehat p_n$ $(\times10^{-5})$
& $95\%$ CI $(\times10^{-5})$ & Time (s)\\
\hline
5 & 3 & 3.0 & $[1.02,8.82]$ & 1.12\\
6 & 6 & 6.0 & $[2.75,13.09]$ & 2.07\\
7 & 0 & 0.0 & $[0,3.84]$ & 4.02\\
8 & 0 & 0.0 & $[0,3.84]$ & 8.27\\
\hline
\end{tabular}
\end{table}

For $n=5$, the exact probability
\[
p_5=\frac{1}{21504}\approx4.65030\times10^{-5}
\]
lies within the reported interval.
For $n=6$, the observed frequency
\[
\widehat p_6=6\times10^{-5}
\]
is below the refined upper bound
\[
B_6^{\mathrm{conj}}
\approx1.52264\times10^{-4}.
\]

For $n=7,8$, no nontrivial stabilizers are observed.
This does not imply that the true probabilities are zero.
In these dimensions, the theoretical bounds are smaller than the upper endpoints of the sampling intervals. Let \(\mathbb E[k_n]\) denote the expected number of samples
with a nontrivial stabilizer.
Indeed,
\[
\mathbb E[k_n]
=
mp_n
\le
mB_n^{\mathrm{conj}},
\]
and therefore
\[
\mathbb E[k_7]\le5.95276\times10^{-5},
\qquad
\mathbb E[k_8]\le1.52613\times10^{-3}.
\]
The absence of observed events is therefore expected at this sample size.



\section{Conclusion}

We studied basis rigidity for finite field inversion under outer affine
transformations, with the AES S-box as the main motivating example. For
AES, we proved that no nontrivial linear change of basis preserves
the coordinate realization. In fact, the AES S-box linear layer already removes the basis redundancy
of pure inversion.

For the general family $A\circ J_n$, we characterized which semilinear
candidates can occur as linear symmetries, determined the
corresponding centralizer dimensions, and derived both
finite dimensional bounds and an asymptotic rigidity result. In particular, for a uniformly random
invertible linear map $A$, the probability that the linear stabilizer is nontrivial is bounded by $2^{-\Omega(n^2)}$.

The computations independently confirm the AES result, the structural
formulas, and the finite dimensional bounds. Overall, the results show
that the basis redundancy of inversion can be completely destroyed by
an outer linear layer and is generically absent in the family considered
here.

\bibliographystyle{plain}
\bibliography{reference}

\end{document}